\RequirePackage{fix-cm}
\documentclass[smallextended]{svjour3}       
\smartqed  
\usepackage{graphicx}
\usepackage{mathrsfs}
\usepackage{amsfonts}
\usepackage{graphicx}
\usepackage{amssymb}
\usepackage{amsmath}
\usepackage{color}
\usepackage{setspace}
\usepackage[misc]{ifsym}

\newtheorem{open problem}{Open Problem}
\journalname{}
\begin{document}

\title{Extended quasi-cyclic constructions of quantum codes and entanglement-assisted
quantum codes}

\titlerunning{Extended quasi-cyclic constructions of quantum codes and entanglement-assisted
quantum codes}        

\author{Jingjie Lv \and Ruihu Li \and Yu Yao}


\institute{
               Jingjie Lv\at
              Department of Basic Sciences, Air Force Engineering
              University\\ Xi'an, 710051, P. R. China\\
              \email{juxianljj@163.com}           
             \and
             \Letter~ Ruihu Li\at
              Department of Basic Sciences, Air Force Engineering
              University\\ Xi'an, 710051, P. R. China\\
              \email{llzsy110@163.com}
\and Yu Yao\at
              Department of Basic Sciences, Air Force Engineering
              University\\ Xi'an, 710051, P. R. China\\
              \email{njjbxqyy@163.com}
              }
\date{Received: date / Accepted: date}

\maketitle

\begin{abstract}
Construction of quantum codes and entanglement-assisted quantum
codes with good parameters via classical codes is an important task
for quantum computing and quantum information. In this paper, by a
family of one-generator quasi-cyclic codes, we provide quasi-cyclic
extended constructions that preserve the self-orthogonality to
obtain stabilizer quantum codes. As for the computational results,
some binary and ternary stabilizer codes with good parameters are
constructed. Moreover, we
 present methods to construct maximal-entanglement
entanglement-assisted quantum codes by means of the class of
quasi-cyclic codes and their extended codes. As an application, some
good maximal-entanglement entanglement-assisted quantum codes are
obtained and their parameters are compared.

\end{abstract}
\section{Introduction}

Quantum error-correcting codes (QECCs, for short) were introduced to
reduce the effects of environmental and operational noise
(decoherence). Reducing the decoherence or controlling the
decoherence to an acceptable level is a key challenge for
researchers. After the pioneering work of Shor and Steane in
\cite{Shor,Steane}, the theory of QECCs has been developed rapidly.
In \cite{Calderbank}, it has been proven that binary QECCs can be
constructed from classical codes over $\mathbb{F}_{2}$ or
$\mathbb{F}_{4}$ with certain self-orthogonal properties. Then,
Ashikhmin et al. \cite{Ashikhmin} generalized these results to
non-binary case . Afterwards, many good QECCs have been constructed
by using classical linear codes over finite fields
\cite{Aly,Ezerman,Gao,Kai,Ketkar,Lv1,Ma,Tonchev,Zhu}. Another
important discovery in the quantum error-correcting area was the
 entanglement-assisted quantum codes (EAQECCs, for short).
The concept of an EAQECC was introduced by Brun et al. \cite{Brun},
which overcame the barrier of the self-orthogonal condition. They
proved that if shared entanglement is available between the sender
and receiver in advance, non-self-orthogonal classical codes can be
used to construct EAQECCs. There are many researchers presented some
constructions of good EAQECCs
\cite{Bowen,Guenda,Lai1,Lai2,Liu,Lu1,Lu2,Qian3,Wang,Wilde1}.
\par
 Quasi-cyclic (QC) codes form an important
class of linear codes with a rich algebraic structure, which is the
generalization of cyclic codes. Many QC codes have improved the
earlier known minimum distances \cite{Daskalov,Siap}. Moreover, QC
codes meet a modified version of the Gilbert-Vashamov (GV) bound
\cite{Kasami,Ling2}.
 Naturally, QC codes can also be applied to construct
QECCs. In \cite{Hagiwara}, Hagiwara et al. studied constructions of
QECCs from QC LDPC codes with a probabilistic method. In 2018,
Galindo et al. \cite{Galindo} used two-generators QC codes that were
dual-containing to construct QECCs. In 2019, Ezerman et al.
\cite{Ezerman} employed QC codes with large Hermitian hulls to
provide QECCs and gained a record-breaking binary $[[31,9,7]]$ QECC.
In \cite{Lv1,Lv2}, we have obtained some good QECCs from
one-generator and two-generators QC codes that are symplectic
self-orthogonal, respectively. It is well-known that extended
constructions of linear error-correcting codes are extremely
effective methods to obtain new codes with good performance, such as
famous generalized Reed-Muller codes, generalized Reed-Solomon codes
and so on. In fact, extended constructions can also be utilized to
construct QECCs and EAQECCs. In \cite{Tonchev}, Tonchev presented
doubling extended constructions and obtained a new optimal binary
$[[28,12,6]]$ QECC, which improved the corresponding lower bounds on
minimum distance at that time. In 2018, Guenda et al. \cite{Guenda}
applied extended constructions to design families of EAQECCs with
good error-correcting performance requiring desirable amounts of
entanglement. However, until now, there are no related extended
constructions of QECCs and EAQECCs from QC codes.

\par
Inspired by the above work, we provide QC extended constructions to
obtain QECCs and maximal-entanglement EAQECCs. This paper is
organized as follows. In Sect. 2, we discuss some preliminary
concepts and propose a family of one-generator QC codes, which are
self-orthogonal with respect to the Hermitian inner product. Sect. 3
provides QC extended constructions that preserve the
self-orthogonality to construct good QECCs. In Sect. 4, we present
methods of constructing maximal-entanglement EAQECCs based on these
QC codes and their extended codes. Sect. 5 concludes this paper.
\par
\section{Preliminaries}
let $\mathbb{F}_{q^2}$ be the finite field with $q^2$ elements where
$q$ is a power of prime $p$. It is obvious that the characteristic
of $\mathbb{F}_{q^2}$ is $p$. Given two vectors $\textbf{u}=(u_0,
u_1, \ldots, u_{n-1})$ and $\textbf{v}=(v_0, v_1, \ldots, v_{n-1})
\in \mathbb{F}_{q^2}^n$, their Hermitian inner product is defined as
$\langle \textbf{u},\textbf{v}\rangle _h=\sum_{i=0}^{n-1}u_i^qv_i$.
Recall that an $[n,k]_{q^{2}}$ linear code $\mathscr{C}$ is a linear
subspace of $\mathbb{F}_{q^2}^{n}$ with dimension $k$. For any
codeword $\textbf{c}\in\mathscr{C}$, the Hamming weight of
$\textbf{c}$ is the number of nonzero coordinates in $\textbf{c}$.
The minimum distance $d$ of a linear code $\mathscr{C}$ equals to
the smallest weight of its nonzero codewords. A generator matrix is
a $k\times n$ matrix whose rows form a basis for $\mathscr{C}$.
Given a linear code $\mathscr{C}\subset \mathbb{F}_{q^2}^{n}$, the
Hermitian dual code of $\mathscr{C}$ is
$\mathscr{C}^{\perp_h}=\{\vec{v} \in \mathbb{F}_{q^2}^{n}|\langle
\vec{u},\vec{v}\rangle _h=0,~\forall \vec{u} \in \mathscr{C} \}$. If
$\mathscr{C}\subset \mathscr{C}^{\perp_h}$, then the code
$\mathscr{C}$ is Hermitian self-orthogonal. The Hermitian hull of
$\mathscr{C}$ is the intersection ${\rm
Hull}_h(\mathscr{C})=\mathscr{C}\cap \mathscr{C}^{\perp_h}$. For a
$k\times n$ matrix $G=(g_{ij})_{k\times n}$ and a vector $\textbf{v}
= (v_1, v_2,\ldots,v_n)$
 over $\mathbb{F}_{q^2}$ (viewed as a $1\times n$ matrix), we define
$G^{q}=(g_{ij}^{q})_{k\times n}$ and $\textbf{v}^{q}=(v_1^{q},
v_2^{q},\ldots,v_n^{q})$. Denote by $A^{\dag}$ and
$\textbf{v}^{\dag}$ the conjugate transpose matrices of $A$ and
$\textbf{v}$. It is well-known that a linear code $\mathscr{C}$ is
Hermitian self-orthogonal if and only if $G G^\dag=\textbf{0}$,
where $G$ is a generator matrix of $\mathscr{C}$ and $\textbf{0}$
denotes the zero matrix.
\par
Let $\mathbb{R}_n=\mathbb{F}_{q^2}[x]/\langle x^n-1\rangle$ and
$\mathscr{C}$ be a $q^2$-ary linear code of length $2n$. For any
codeword
$\textbf{c}=(c_{0},c_{1},\ldots,c_{n-1},c_{n},c_{n+1},\ldots,c_{2n-1})\in
\mathscr{C}$, define $\psi(\textbf{c})=(
c_{n-1},c_{0},\ldots,c_{n-2},c_{2n-1}, c_{n},\ldots,c_{2n-2}).$
 If $\mathscr{C}=\psi(\mathscr{C})$, then
we call the code $\mathscr{C}$ a quasi-cyclic (QC) code of length
$2n$ and index 2 over $\mathbb{F}_{q^2}$. Note that a vector
$\textbf{c}=(a_{0},\ldots,a_{n-1},b_{0}, \ldots,b_{n-1}) \in
\mathbb{F}_{q^2}^{2n}$ can be identified with $(a(x),b(x))\in
\mathbb{R}_{n}^{2}$, where $a(x)=a_{0}+a_{1}x+\cdots+a_{n-1}x^{n-1}
$ and $b(x)=b_{0}+b_{1}x+\cdots+b_{n-1}x^{n-1}.$ Further,
$\psi(\textbf{c})$ corresponds to $(xa(x),xb(x))$ in
$\mathbb{R}_{n}^{2}$. Therefore, a QC code $\mathscr{C}$ of index 2
can be viewed algebraically as an $\mathbb{R}_n$-submodule of
$\mathbb{R}_n^{2}$. If $\mathscr{C}$ is generated by
$\mathcal{G}(x)=(g_{1}(x),g_{2}(x))\in\mathbb{R}_{n}^2$, then
$\mathscr{C}$ is a one-generator QC code of index 2. Define
$g(x)={\rm gcd}(g_{1}(x),g_{2}(x),x^{n}-1)$ and
$h(x)=(x^{n}-1)/g(x).$ Polynomials $g(x)$ and $h(x)$ are the
generator polynomial and the parity-check polynomial of
$\mathscr{C}$, respectively. Further, the dimension of $\mathscr{C}$
is ${\rm deg}(h(x))$. One can refer to \cite{Seguin1} for more
details.
\par
Attached to polynomial $f(x)=f_{0}+f_{1}x+\cdots+f_{n-1}x^{n-1}\in
\mathbb{R}_{n}$, we define $\bar{f}(x)=f_{0}+ f_{n-1}x+
f_{n-2}x^{2}+\cdots+ f_{1}x^{n-1}$ and$
f^{q}(x)=f_{0}^{q}+f_{1}^{q}x+\cdots+ f_{n-1}^{q}x^{n-1}.$ In
addition, if $f(x)\cdot h(x)=x^n-1$, then
 $f^{\bot}(x)=x^{{\rm deg}(h(x))}h(\frac{1}{x}).$
In the following, a class of one-generator QC codes that are
Hermitian self-orthogonal is introduced, whose some properties can
be obtained from \cite{Galindo}.
\par
\begin{definition}
Let $\mathscr{C}_{q^2}(f,g)$ be a QC code over $\mathbb{F}_{q^2}$ of
length $2n$ generated by $(g(x),f(x)g(x))$, where $f(x)$ and $g(x)$
are polynomials in $\mathbb{R}_{n}$ such that $g(x)$ divides
$x^n-1$.
\end{definition}
\par

\begin{lemma}(\cite{Galindo}, Proposition 14)
The Hermitian dual code $\mathscr{C}^{\bot_{h}}_{q^2}(f,g)$ of the
QC code $\mathscr{C}_{q^2}(f,g)$ over $\mathbb{F}_{q^2}$ is
generated by the pairs $({g^{\bot q}}(x),0)$ and
$(-\bar{f}^q(x),1)$.
\end{lemma}
\par

\begin{lemma}(\cite{Galindo}, Proposition 15)
A sufficient condition for $\mathscr{C}_{q^2}(f,g)$ to be contained
in its Hermitian dual code $\mathscr{C}^{\bot_{h}}_{q^2}(f,g)$ is
${g^{\bot q}}(x) \mid g(x).$
\end{lemma}
\par
Note that if a linear code $\mathscr{C}$ is self-orthogonal, then
the dual dimension is larger than or equal to that of $\mathscr{C}$.
Since $\mathscr{C}^{\bot_h}=(\mathscr{C}^{\bot})^q$, we conclude
that $\mathscr{C}^{\bot_h}$ and $\mathscr{C}^{\bot}$ have the same
weight distribution, where $\mathscr{C}^{\bot}$ denotes the usual
Euclidean dual code of $\mathscr{C}$. In order to obtain the exact
Hermitian dual distance of $\mathscr{C}$, for simplicity, we can
firstly calculate the weight distribution of $\mathscr{C}$, and then
apply the following well-known MacWilliams equation.
\begin{theorem}\cite{MacWilliams}
If $\mathscr{C}$ is an $[n,k,d]$ linear code over $\mathbb{F}_{q}$
with weight enumerator
$$W_{\mathscr{C}}(x,y)=\sum_{i=0}^{n}A_{i}x^{n-i}y^i,$$
where $A_{i}$ denotes the number of codewords in $\mathscr{C}$ with
Hamming weight $i$. Then, the weight enumerator of the Euclidean
dual code $\mathscr{C}^{\bot}$ is given by
$$W_{\mathscr{C}^{\bot}}(x,y)=q^{-k}W_{\mathscr{C}}(x+(q-1)y,x-y).$$
\end{theorem}
\section{Extended quasi-cyclic constructions of quantum codes}
 Recall that a $q$-ary quantum error-correcting code (QECC)
 of length
 $n$ is a $K$-dimensional subspace of the $q^{n}$-dimensional Hilbert
space $(\mathbb{C}^{q})^{\otimes n}$, where $\mathbb{C}$ denotes the
complex field. If $K = q^k$, then the QECC is represented by
$[[n,k,d]]_q$, where $d$ is the minimum distance. Just as the
classical case, one of the main problems in quantum error correction
is to construct QECCs with good parameters. When fixing the code
length $n$ and dimension $k$, we expect to gain a big minimum
distance $d$. Conversely, when the minimum distance $d$ is equal, we
want the code rate $\frac{k}{n}$ to be greater. Available in
\cite{Grassl1}, there is a database of best known
 binary QECCs. For ternary QECCs, code tables \cite{Edel} are
kept online by Edel according to their explicit constructions. In
order to evaluate the superiority of QECCs, Feng et al. \cite{Feng}
presented a quantum Gilbert-Vashamov (GV) bound as follows, which is
closely related to the size of the finite field.
\begin{theorem} (\cite{Feng},~quantum
Gilbert-Vashamov bound) Let $n>k\geq2$ with $n\equiv k$~({\rm mod}
2) and $d\geq2$. Then there exists a pure stabilizer QECC with
parameters $[[n,k,d]]_q$ if the inequality
$$\frac{q^{n-k+2}-1}{q^2-1}>\sum_{i=1}^{d-1}(q^2-1)^{i-1}\binom{n}{i}$$
is satisfied.
\end{theorem}
\par
 One can check that almost all the QECCs meet this bound. If not, these
codes usually have particularly good parameters. It is generally
known that there exists an important connection between QECCs and
classical Hermitian self-orthogonal linear codes from the following
theorem.
\begin{theorem} \cite{Ashikhmin}
A Hermitian self-orthogonal $[n,k]_{q^{2}}$ linear code
$\mathscr{C}$ such that there are no vectors of weight less than $d$
in $\mathscr{C}^{\perp_h}\backslash \mathscr{C} $ yields a QECC with
parameters $[[n,n-2k,d]]_{q}$.
\end{theorem}
\par
Via the Hermitian self-orthogonal quasi-cyclic codes introduced in
Definition 1, next, we will provide our QC extended constructions
that preserve the self-orthogonality.
\begin{proposition}
Assume that $g(x)$ and $f(x)$ are polynomials in $\mathbb{R}_{n}$
satisfying ${g^{\bot q}}(x) \mid g(x)$, then the QC code
$\mathscr{C}_{q^2}(f,g)=[2n,n-{\rm deg}(g(x)),d]_{q^{2}}$ with
generator matrix $G=(G_1,G_2)$ is Hermitian self-orthogonal. Let
$\mathscr{C}_i$ $(i=1,2)$ be a linear code generated by $G_i$. If
there exists a codeword $x^{(i)}\in \mathscr{C}^{\bot_{h}}_{i}$ such
that $\langle
x^{(i)},x^{(i)}\rangle _h=p-1$, where $p$ is the characteristic of $\mathbb{F}_{q^2}$. Then\\
 (i) The code $\mathscr{C}^{'}$ with generator matrix
 \begin{equation*}
 \begin{small}
G^{'}=\left(
 \begin{array}{ccc|ccc|c}
 &     & & &   & & 0\\
 &G_{1}& & &G_{2}&   &  \vdots\\
 &     & & &   & & 0\\
 \hline
 &x^{(1)}&& &0\cdots0&&1
 \end{array}
\right)
\end{small}
\end{equation*}
is a Hermitian self-orthogonal $[2n+1, n-{\rm deg}(g(x))+1]$ linear
code with Hermitian dual distance
$$d^{\bot_{h}}\leq d({\mathscr{C}^{'\bot_{h}}}) \leq d^{\bot_{h}}+1,$$
where $d^{\bot_{h}}$ denotes the Hermitian dual distance of
$\mathscr{C}_{q^2}(f,g)$.\\
 (ii) The code $\mathscr{C}^{''}$ with generator matrix
  \begin{equation*}
  \begin{small}
G^{''}=\left(
 \begin{array}{ccc|ccc|cc}
 &     & & &   & & 0&0\\
 &G_{1}& & &G_{2}&  &  \vdots&\vdots \\
 &     & & &   & & 0&0\\
 \hline
 &x^{(1)}&& &0\cdots0&&1&0\\
\hline
  &0\cdots0&& &x^{(2)}&&0&1\\
 \end{array}
\right)
\end{small}
\end{equation*}
is a Hermitian self-orthogonal $[2n + 2, n-{\rm deg}(g(x)) +2]$
linear code with Hermitian dual distance
$$d^{\bot_{h}}\leq d({\mathscr{C}^{''\bot_{h}}}) \leq d^{\bot_{h}}+2,$$
where $d^{\bot_{h}}$ denotes the Hermitian dual distance of
$\mathscr{C}_{q^2}(f,g)$.
\end{proposition}
\begin{proof}
Obviously, the linear codes $\mathscr{C}^{'}$ and $\mathscr{C}^{''}$
have parameters $[2n + 1, n-{\rm deg}(g(x)) + 1]$ and $[2n + 2,
n-{\rm deg}(g(x)) + 2]$, respectively. A simple computation shows
that
 \begin{small}
 \begin{equation*}
\begin{split}
G^{'}{G^{'}}^\dag&=\left(
 \begin{array}{ccc|ccc|c}
 &     & & &   & & 0\\
 &G_{1}& & &G_{2}&   &  \vdots\\
 &     & & &   & & 0\\
 \hline
 &x^{(1)}&& &0\cdots0&&1
 \end{array}
\right)  \left(
 \begin{array}{ccc|ccc|c}
 &     & & &   & & 0\\
 &G_{1}& & &G_{2}&   &  \vdots\\
 &     & & &   & & 0\\
 \hline
 &x^{(1)}&& &0\cdots0&&1
 \end{array}
\right)^{\dag}\\
 &=\left(
 \begin{array}{ccc|ccc}
 &     & & &   & \\
 &G_{1} G_{1}^{\dag}+G_{2}G_{2}^{\dag}&  &G_{1} {x^{(1)}}^{\dag}&   &  \\
 &     & & &   & \\
 \hline
 &x^{(1)} G_{1}^{\dag}&&x^{(1)}{x^{(1)}}^{\dag}+1 &&
 \end{array}
\right)
  \end{split}
\end{equation*}
  \end{small}
and
 \begin{small}
 \begin{equation*}
 \begin{split}
G^{''}{G^{''}}^\dag&=\left(
 \begin{array}{ccc|ccc|cc}
 &     & & &   & & 0&0\\
 &G_{1}& & &G_{2}&   &  \vdots&\vdots\\
 &     & & &   & & 0&0\\
 \hline
 &x^{(1)}&& &0\cdots0&&1&0\\
\hline
  &0\cdots0&& &x^{(2)}&&0&1\\
 \end{array}
\right) \left(
 \begin{array}{ccc|ccc|cc}
 &     & & &   & & 0&0\\
 &G_{1}& & &G_{2}&   &  \vdots&\vdots\\
 &     & & &   & & 0&0\\
 \hline
 &x^{(1)}&& &0\cdots0&&1&0\\
\hline
  &0\cdots0&& &x^{(2)}&&0&1\\
 \end{array}
\right)^{\dag}\\
 &=\left(
 \begin{array}{ccc|ccc|c}
 &     & & &   & \\
 &G_{1} G_{1}^{\dag}+G_{2}G_{2}^{\dag}&  &G_{1} {x^{(1)}}^{\dag}&   &  &G_{2} {x^{(2)}}^{\dag}\\
 &     & & &   & \\
 \hline
 &x^{(1)} G_{1}^{\dag}&&x^{(1)}{x^{(1)}}^{\dag}+1 &&&0
 \\
  \hline
 &x^{(2)} G_{2}^{\dag}&&0 &&&x^{(2)}{x^{(2)}}^{\dag}+1
 \end{array}
\right).
  \end{split}
\end{equation*}
 \end{small}
 Since the QC code
$\mathscr{C}_{q^2}(f,g)$ is Hermitian self-orthogonal, then
$GG^{\dag}=G_{1} G_{1}^{\dag}+G_{2}G_{2}^{\dag}=\textbf{0}$, where
\textbf{0} is the zero matrix. If there exists a codeword
$x^{(i)}\in \mathscr{C}^{\bot_{h}}_{i} (i=1,2)$ such that $\langle
x^{(i)},x^{(i)}\rangle _h=p-1$, then it is easy to see that
$G^{'}{G^{'}}^\dag$ and $G^{''}{G^{''}}^\dag$ are both zero
matrices. It is equivalent to say that linear codes
$\mathscr{C}^{'}$ and $\mathscr{C}^{''}$ are both Hermitian
self-orthogonal. Further, as every $d^{\bot_{h}}-1$ columns of $G$
are linearly independent, then every $d^{\bot_{h}}-1$ columns of
$G^{'}$ and $G^{''}$ are obviously linearly independent. It follows
that $d^{\bot_{h}}\leq d({\mathscr{C}^{'\bot_{h}}}) \leq
d^{\bot_{h}}+1$ and $d^{\bot_{h}}\leq d({\mathscr{C}^{''\bot_{h}}})
\leq d^{\bot_{h}}+2$, where $d^{\bot_{h}}$ denotes the Hermitian
dual distance of $\mathscr{C}_{q^2}(f,g)$.
\end{proof}
\par
By Theorem 3 and Proposition 1, we have the following result
directly.
\begin{theorem}
With the above notations, let $\mathscr{C}_{q^2}(f,g)$ be a
self-orthogonal QC code $\mathscr{C}_{q^2}(f,g)$ with respect to the
Hermitian inner product. Then it provides two QECCs with parameters
$[[2n+1,2{\rm deg}(g(x))-1, d({\mathscr{C}^{'\bot_{h}}})]]_q$ and
$[[2n+2,2{\rm deg}(g(x))-2, d({\mathscr{C}^{''\bot_{h}}})]]_q$,
respectively. Moreover, $d^{\bot_{h}}\leq
d({\mathscr{C}^{'\bot_{h}}}) \leq d^{\bot_{h}}+1$ and
$d^{\bot_{h}}\leq d({\mathscr{C}^{''\bot_{h}}}) \leq
d^{\bot_{h}}+2$, where $d^{\bot_{h}}$ denotes the Hermitian dual
distance of $\mathscr{C}_{q^2}(f,g)$.
\end{theorem}
\par
In the following, we will construct some good QECCs over small
finite fields $\mathbb{F}_{2}$ and $\mathbb{F}_{3}$ according to
Theorem 4. We compute it
 by the algebra system Magma \cite{Bosma}. Let
$\omega$ and $\xi$ be primitive elements of $\mathbb{F}_{4}$ and
$\mathbb{F}_{9}$. For simplicity, elements $0,1,\omega,\omega^2$ in
$\mathbb{F}_{4}$ and
$0,1,\xi,\xi^{2},\xi^{3},\xi^{4},\xi^{5},\xi^{6},\xi^{7}$ in
$\mathbb{F}_{9}$ are represented by 0,1,2,3 and 0,1,2,3,4,5,6,7,8,
respectively.

\begin{example}
Assume that $q=2$ and $n=15$. Consider the following polynomials in
$\mathbb{F}_{4}[x]/\langle x^{15}-1\rangle$,
\begin{equation*}
\begin{split}
g(x)=x^9 + 3x^8 + x^7 + x^5 + 3x^4 + 2 x^2 + 2 x + 1,
\end{split}
\end{equation*}
$$f(x)=2 x^3+2 x^2+2 x+1.$$
Since $g(x)\mid (x^{15}-1)$ and ${g}^{\bot q}(x) \mid g(x)$, then by
Lemma 2, the QC code $\mathscr{C}_{4}(f,g)$ is Hermitian
self-orthogonal. Select a codeword $x^{(1)}=(1,3, 2, 1, 3,
 2, 1, 3, 2, 1, 3,
\\2, 1, 3, 2)\in \mathscr{C}^{\bot_{h}}_{1}$. According to
Proposition 1 (i), we can construct a Hermitian self-orthogonal
linear code $\mathscr{C}^{'}$ with generator matrix
\begin{equation*}
\begin{tiny}
 G^{'}=\left(
\begin{array}{ccccccccccccccccccccccccccccccc}
 1 &2& 2& 0 &3 &1& 0& 1 &3& 1& 0 &0& 0& 0& 0&1& 0& 3 &2 &3 &3 &3 &2 &3& 2& 1& 3& 2 &0 &0& 0 \\
 0 &1 &2& 2 &0& 3& 1 &0 &1& 3& 1& 0& 0& 0& 0&0& 1& 0& 3 &2 &3 &3 &3 &2 &3& 2& 1& 3& 2 &0 &0 \\
 0 &0 &1& 2& 2 &0& 3& 1& 0& 1& 3& 1& 0& 0& 0&0& 0& 1& 0& 3 &2 &3 &3 &3 &2 &3& 2& 1& 3& 2 &0  \\
 0 &0& 0& 1& 2& 2 &0& 3&1& 0& 1& 3& 1& 0& 0&2& 0& 0& 1& 0& 3 &2 &3 &3 &3 &2&3& 2& 1& 3 &0  \\
 0 &0& 0& 0& 1&2& 2& 0& 3& 1& 0& 1& 3& 1& 0 &3 &2& 0& 0& 1& 0& 3 &2 &3 &3&3 &2 &3& 2& 1& 0\\
 0 &0&0& 0& 0& 1&2& 2& 0& 3& 1& 0& 1& 3& 1  &1&3 &2& 0& 0& 1& 0& 3 &2 &3&3 &3 &2 &3& 2&  0\\
 1 &3& 2& 1& 3& 2& 1& 3& 2& 1& 3& 2& 1& 3& 2& 0& 0& 0 &0& 0&0 &0&0 &0 &0 &0 &0& 0& 0& 0&1\\
\end{array}
\right).
\end{tiny}
\end{equation*}
Using algebraic software Magma \cite{Bosma}, we have that
$\mathscr{C}^{'}$ is a $[31,7,16]_{4}$ linear code
 and its weight enumerator is
\begin{equation*}
\begin{split}
0^1 16^3 18^{630}20^{2520}22^{3900}24^{5400}26^{3150}28^{780}.
\end{split}
\end{equation*}
By the MacWilliams equation, we can gain that the weight enumerator
of $\mathscr{C}^{'\bot_{h}}$ is
\begin{equation*}
\begin{split}
&0^15^{2709}6^{33789}7^{352635}8^{3146895}9^{24208470}10^{
159955686}11^{915334434}12^{4577489490}\\
&13^{20070644055}14^{77414126895}15^{263209977249}16^{789626267391}17^{2090205270180}\\
&18^{4877070505860}19^{10011021610380}20^{18019505816172}21^{28316806886643}22^{38613305033355}\\
&23^{45329278307085}24^{45328597815825}25^{38076647339430}26^{26360313433398}27^{14644850099250}\\
&28^{6276277886370}29^{1947832175745}30^{389563102377}31^{37699888887}.\\
\end{split}
\end{equation*}
Hence, the Hermitian dual distance of $\mathscr{C}^{'}$ is 5.  By
Theorem 4, an optimal QECC with parameters $[[31,17,5]]_{2}$ will be
provided. Moreover, according to the propagation rule in
\cite{Calderbank}, there will also exists a QECC with parameters
$[[32,17,5]]_{2}$. By comparison with Grassl's code tables
\cite{Grassl1}, we note that our codes are better than the
best-known $[[31,17,4]]_{2}$ and $[[32,17,4]]_{2}$ QECCs,
respectively. Hence, our QECCs break the current records.
\end{example}

\begin{example}
Write $q=3$ and $n=10$.  Choose the following polynomials in
$\mathbb{F}_{9}[x]/\langle x^{10}-1\rangle$,
\begin{equation*}
\begin{split}
g(x)=x^6 + 7x^5 + 5x^4 + x^2 + 3x +5,~~ f(x)=x^3+2x^2+5x+1.
\end{split}
\end{equation*}
By Lemma 2, the QC code $\mathscr{C}_{9}(f,g)$ is Hermitian
self-orthogonal. Select codewords $x^{(1)}=( 1, 1, 8, 2, 1, 2, 2, 6,
0, 1) \in\mathscr{C}^{\bot_{h}}_{1}$ and $x^{(2)}=(1 , 7, 3, 8, 5,
7, 7, 0, 3, 2)\in\mathscr{C}^{\bot_{h}}_{2}$. By Proposition 1 (ii),
a Hermitian self-orthogonal linear code $\mathscr{C}^{''}$ can be
constructed, whose generator matrix is given as follows
\begin{equation*}
 \begin{footnotesize}
 G^{''}=\left(
\begin{array}{ccccc ccccc ccccc ccccc cc}
5& 3& 1& 0& 5& 7& 1& 0& 0& 0& 5& 8& 2& 7& 6& 4& 5& 2& 5& 1& 0& 0\\
 0& 5& 3& 1& 0& 5& 7& 1& 0& 0& 1& 5& 8& 2& 7& 6& 4& 5& 2& 5& 0& 0\\
 0& 0& 5& 3& 1& 0& 5&7 &1 &0 &5 &1 &5 &8 &2 &7 &6 &4 &5 &2& 0& 0\\
 0& 0& 0& 5& 3& 1& 0& 5& 7& 1&2& 5& 1& 5& 8& 2& 7& 6& 4& 5& 0& 0\\
1 &1& 8& 2& 1& 2& 2& 6& 0& 1& 0& 0& 0& 0&
0 &0& 0& 0& 0& 0& 1& 0\\
 0& 0& 0& 0& 0& 0& 0& 0& 0& 0& 1& 7& 3& 8& 5& 7& 7& 0&
3 &2& 0& 1\\
\end{array}
\right).
 \end{footnotesize}
\end{equation*}
Calculate that $\mathscr{C}^{''}$ is a $[22,6,10]_{9}$ linear code
and its weight enumerator is
\begin{equation*}
\begin{split}
0^1 10^{16} 12^8 13^{80}14^{624}15^{3376}16^{11192}17^{
32856}18^{71520}19^{118336}20^{142128}21^{112664}22^{38640}.
\end{split}
\end{equation*}
By the MacWilliams equation, we obtain that $\mathscr{C}^{'}$ is a
$[22,16,5]_{9}$ linear code. One can check that it is the best-known
classical code according to Grassl's code tables \cite{Grassl1}. By
Theorem 4, there exists a QECC with parameters $[[22,10,5]]_{3}$. To
testify the superiority of the code, we find that our code exceeds
the quantum GV bounds. In \cite{Ezerman}, the authors gave a QECC
with parameters $[[22,8,5]]_{3}$. Obviously, the code rate of our
QECC is higher.
\end{example}
\begin{example}
Now set $q=2$ and $n=51$. Define polynomials
\begin{equation*}
\begin{split}
g(x)=&x^{35}+2 x^{34} +3x^{33} + 3x^{32} + 3x^{30} + 2 x^{26} +
x^{25} + 2 x^{24} + x^{23} +
 2 x^{22}\\
 &+2 x^{20} + 3x^{18} + 2 x^{17} + 3x^{15} + 2 x^{13} + 2 x^{12} + x^{11} +
 x^{10} +3x^9\\
 &+2 x^5 + x^3 + 2 x^2 + 2 x + 2,\\
\end{split}
\end{equation*}
$$f(x)=2 x^{15}+ x^{14}+x^{13}+ x^{12}+x^{11}+2 x^{10}
 + x^9+x^8 +x^7 + x^6 + x^5 +x^4 + x^3 + x^2 + x + 1$$
 in $\mathbb{F}_{4}[x]/\langle x^{51}-1\rangle$. By Lemma 2,
 the QC code $\mathscr{C}_{9}(f,g)$ is Hermitian self-orthogonal.
Choose a codeword $x^{(1)}=(1, 2, 0, 0, 0, 3, 2,0, 2, 3, 3, 0, 2, 3,
1, 3, 0,2, \\2, 3, 2, 3, 0, 3, 1, 3, 2, 1, 2, 2, 3, 0, 2, 1, 3, 3,
1, 0, 2, 0, 3, 0, 1, 2, 2, 0, 1, 3, 3, 0, 0)\in
\mathscr{C}^{\bot_{h}}_{1}$. By Proposition 4 (i), we obtain a
Hermitian self-orthogonal $[103,17,38]_{4}$ linear code, whose
weight enumerator is given as follows
\begin{equation*}
\begin{split}
&0^1 38^3 48^6 50^{96}52^{1971}54^{14862}56^{92127}58^{
551322}60^{2784441}62^{11959407}64^{43955487}66^{136538139}\\
&68^{359079711}70^{
796085910}72^{1480101177}74^{2293531833}76^{2941096230}78^{
3093630249}80^{2642685339}\\
& 82^{1811850639}84^{982791951}86^{
413828565}88^{132236040}90^{31182537}92^{5228613}94^{597612}96^{43263}\\
&98^{1626}100^{27}.
\end{split}
\end{equation*}
By the MacWilliams equation, the Hermitian dual distance of
$\mathscr{C}^{'}$ is 7. By Theorem 3, a QECC with parameters
$[[103,69,7]]_{2}$ can be constructed, which surpasses the
best-known $[[103,69,6]]_2$ QECC at now \cite{Grassl1}.
\end{example}
\par
In the following, four tables will be given to illustrate that many
good QECCs can be obtained by our extended QC constructions. Here we
just give some good QECCs over small finite fields $\mathbb{F}_{2}$
and $\mathbb{F}_{3}$ via the extended construction provided in
Proposition 1 (i). Tables 1 and 3 contain some Hermitian
self-orthogonal linear codes $\mathscr{C}^{'}$ over $\mathbb{F}_{4}$
and $\mathbb{F}_{9}$. These codes are used to construct good binary
and ternary QECCs in Tables 2 and 4, respectively. In Table 2, some
the best-known or optimal binary QECCs with length less than or
equal to 127 are provided, most of which have different weight
distributions with the best-known QECCs in Grassl's code
tables\cite{Grassl1}. In Table 4, we construct some good ternary
QECCs, which all exceed the quantum GV bounds and have higher code
rate than QECCs available in \cite{Edel}. For simplicity, we write
coefficients of polynomials in ascending order to denote
polynomials. The exponents of the elements indicate the number of
the consecutive same elements. For example, the polynomial
$1+\xi^{4}x^2+x^3+x^4$ over $\mathbb{F}_{9}$ is represented by
$1051^2$.

\begin{table}
   \caption{Hermitian self-orthogonal extended QC codes $\mathscr{C}^{'}$ over $\mathbb{F}_4$.}
   \begin{center}
   \begin{spacing}{1.6}
   \begin{scriptsize}
   \begin{tabular}{ccc}

  \hline  \hline
   $n$ &$f(x)$, $g(x)$ and $x^{(1)}$ over $\mathbb{F}_{4}$ & Codes $\mathscr{C}^{'}$\\
   \hline
7& $12$,\quad$101^3$,\quad$(13)^23^21$   &$[15,4,8]_{4}$ \\
    \hline

$17$&$3^31$,\quad$132^20^22^231$,\quad$13^210^42^230(21)^20$ &$[35,9,14]_{4}$\\
\hline
$23$&$1^623$,\quad$10(100)^21^5$,\quad$10232^20^3313020^232^20^33$&$[47,12,20]_{4}$\\
\hline $29$&$1^9212$,\quad$12(331)^2(133)^221$,\quad$1021^30^410
3^2203^201^2(21)^2131^2$ & $[59,15,24]_{4}$\\
\hline
$31$&$1^73$,\quad$1^701^2(01)^210^2101^3$,\quad$10^2132301^2013101^223^2030201^2313^2$ & $[63,11,24]_{4}$\\
\hline
$31$&$1^{12}212$,\quad$10^31^40^410^21^2$,\quad$(10)^2020^213^20^213031^20212012^21321$ &$[63,16,22]_{4}$\\
\hline
$37$&$1^{14}2013$,\quad$12^2020132^4310202^21$, &$[75,19,26]_{4}$\\
 & $(10^2)^23^22^3310^21^22^21310^23230102120102$&\\
\hline
$39$&$1^{14}3203$,\quad$121^3302^312^21302^213^21$, & $[79,19,32]_{4}$\\
& $1^203^32^43230312313(23)^203(20)^23^21210^33$&\\
\hline
$41$& $1^72^21$,\quad$131210(31)^21012^23^22^2101^2313012131$, &$[83,11,30]_{4}$\\
 & $130^3232012^212012^23^220(3101)^221203^203020$&\\
\hline
$55$&$1^{10}2$,\quad$130132^230203^22131(0^22)^23^310^210210(31)^232^301$, & $[111,13,46]_{4}$\\
 & $(13)^221^30^21302^332102^21^2202(30)^32^43^203^2210^320232^203$&\\
\hline
$63$&$1^{12}212$,\quad$31231^232030^33^2012^31213^202302^23213(10)^231^20^21^2210121$, & $[127,13,52]_{4}$\\
& $1^20232^3321310^22131210^2201^3030121^42^412^401(32)^2013^201^202^3$&\\
\hline
$63$&$1^{12}212$,\quad$31^532^30^2231^23203^21^330^42^212^201203123^21^23231$, &$[127,16,50]_{4}$\\
 & $10^2303(20)^2121^32^3121^202^31^23123(12)^323213231323^202^3(30)^212310$&\\
\hline \hline
   \end{tabular}
   \end{scriptsize}
   \end{spacing}
\end{center}
\end{table}

\begin{table}
   \caption{The best-known binary QECCs from extended QC codes $\mathscr{C}^{'}$.}
   \begin{center}
   \begin{spacing}{1.6}
   \begin{scriptsize}
   \begin{tabular}{ccc}

  \hline\hline
  $\mathscr{C}_{4}(f,g)$& $\mathscr{C}_{4}^{\bot_{h}}(f,g)$& Our QECCs\\
   \hline
$[15, 4, 8]_{4}$&$[15, 11, 3]_{4}$&$[[15, 7, 3]]_{2}$ \\
$[35, 9, 14]_{4}$& $[35, 26, 5]_{4}$ &$[[35, 17, 5 ]]_{2}$\\
$[47, 12, 20]_{4}$& $[47, 35, 6 ]_{4}$ &$[[47, 23, 6 ]]_{2}$\\
$[59, 15, 24 ]_{4}$& $[59, 44, 7 ]_{4}$ &$[[59, 29, 7 ]]_{2}$\\
$[63, 11, 24]_{4}$& $[63, 52, 5 ]_{4}$ &$[[63, 41, 5 ]]_{2}$\\
$[63, 16, 22]_{4}$& $[63,47,7]_{4}$ &$[[63,31,7]]_{2}$\\
$[75, 19, 26 ]_{4}$& $[75, 56, 8]_{4}$ &$[[75, 37, 8 ]]_{2}$\\
$[79, 19, 32 ]_{4}$& $[79,60,8]_{4}$ &$[[79,41,8]]_{2}$\\
$[83, 11, 30]_{4}$& $[83,72,5]_{4}$ &$[[83,61,5]]_{2}$\\
$[111, 13, 46 ]_{4}$& $[111,98,5]_{4}$ &$[[111,85,5]]_{2}$\\
$[127, 13, 52]_{4}$& $[127,114,5]_{4}$ &$[[127,101,5]]_{2}$\\
$[127, 16, 50 ]_{4}$& $[127,111,6]_{4}$ &$[[127,95,6]]_{2}$\\
\hline\hline
   \end{tabular}
   \end{scriptsize}
   \end{spacing}
\end{center}
\end{table}

\begin{table}
   \caption{Hermitian self-orthogonal extended QC codes $\mathscr{C}^{'}$ over $\mathbb{F}_9$.}
   \begin{center}
   \begin{spacing}{1.6}
   \begin{scriptsize}
   \begin{tabular}{ccc}
\hline  \hline
   $n$ &$f(x),g(x),x^{(1)}$ over $\mathbb{F}_{9}$ &Codes $\mathscr{C}^{'}$\\
   \hline
11&$12486$,\quad$15^3101$,\quad$126245487^3$ &$[23,6,12]_{9}$\\
\hline
$17$&$1^45121$,\quad$5215371561$,\quad$1^23680^21726823472$ &$[35,9,16]_{9}$\\
\hline
$23$&$1^8212$,\quad$150(51)^2(10)^201$, &$[47,12,23]_{9}$\\
&$18452373054381^26383157^2$&\\
 \hline
$35$&$1^621$,\quad$5208270^2(75)^2540276513148^2731$, & $[71,9,26]_{9}$\\
 & $1050^22676308^2316^202384^20^373487^280^2$&\\
\hline
$41$&$1^206$,\quad$583540135073452^26126526^218730175081741$, & $[83,5,39]_{9}$\\
 & $1743516718^230141^2786273^281(28)^2245^28631^242$&\\
 \hline
$65$&$1^921$,\quad$173681^2057206^22847641684587643^2746825^280^21340275868531$, & $[131,12,59]_{9}$\\
 & $17361^28^225412708058626127^2805(26)^212^27080(08)^2128642857381682^264214$&\\
\hline \hline
   \end{tabular}
   \end{scriptsize}
   \end{spacing}
\end{center}
\end{table}

\begin{table}
   \caption{Ternary QECCs from extended QC codes $\mathscr{C}^{'}$.}
   \begin{center}
   \begin{spacing}{1.6}
   \begin{scriptsize}
   \begin{tabular}{cccccc}

  \hline\hline
  $\mathscr{C}_{9}(f,g)$& $\mathscr{C}_{9}^{\bot_{h}}(f,g)$&Our QECCs &Rate &QECCs \cite{Edel}&Rate \\
   \hline
$[23,6,12]_{9}$& $[23,17,5]_{9}$&$[[23,11,5]]_{3}$&0.478&$[[21,7,5]]_{3}$&0.333\\
$[35,9,16]_{9}$& $[35,26,6]_{9}$ &$[[35,17,6]]_{3}$&0.486&$[[26,11,6]]_{3}$&0.423\\
$[47,12,23]_{9}$& $[47,35,7]_{9}$ &$[[47,23,7]]_{3}$&0.489&$[[52,25,7]]_{3}$&0.481\\
$[71,9,26]_{9}$& $[71,62,5]_{9}$ &$[[71,53,5]]_{3}$&0.746&$[[65,43,5]]_{3}$&0.662\\
$[83,5,39]_{9}$& $[83,78,4]_{9}$ &$[[83,73,4]]_{3}$&0.880&$[[81,71,4]]_{3}$&0.877\\
$[131,12,59]_{9}$& $[131,119,6]_{9}$ &$[[131,107,6]]_{3}$&0.817&$[[140,106,6]]_{3}$&0.757\\
\hline\hline
   \end{tabular}
   \end{scriptsize}
   \end{spacing}
\end{center}
\end{table}

\section{Quasi-cyclic constructions of entanglement-assistant quantum codes}
Entanglement-assistant quantum error-correcting codes (EAQECCs) can
be regarded as generalized QECCs, which can break the
self-orthogonal conditions. An $[[n,k,d;c]]_{q}$ EAQECC encodes $k$
logical qubits into $n$ physical qubits with the help of $c$ copies
of entangled ebits. In particular, if $c=0$, then the EAQECC is a
standard stabilizer QECC. Similar to the QECCs, EAQECCs can also be
constructed by classical linear codes in the following theorem.
\begin{theorem}(\cite{Wilde})
If $\mathscr{C}$ is an $[n,k,d]_{q^{2}}$ classical code over
$\mathbb{F}_{q^{2}}$ with parity check matrix $H$, then
$\mathscr{C}^{\bot_{h}}$ stabilizes an EAQECC with parameters
$[[n,2k-n+c,d;c]]_{q}$, where $c={\rm Rank}(HH^{\dag})$ is the
number of entangled ebits required.
\end{theorem}
\par
In 2018, Guenda et al. \cite{Guenda} established a relation between
the required number of entangled ebits and the dimension of the
Hermitian hull of a classical linear code.
\begin{theorem}(\cite{Guenda})
Let $\mathscr{C}$ be a classical $[n,k,d]_{q^{2}}$ code with parity
check matrix $H$ and  generator matrix $G$. Then ${\rm
Rank}(HH^{\dag})$ and ${\rm Rank}(GG^{\dag})$ are independent of $H$
and $G$ so that
$${\rm Rank}(HH^{\dag})=
n-k-{\rm dim}({\rm Hull}_h(\mathscr{C}))=n-k-{\rm dim}({\rm
Hull}_h(\mathscr{C}^{\bot}_h )),$$ and
$${\rm Rank}(GG^{\dag})=
k-{\rm dim}({\rm Hull}_h(\mathscr{C}))=k-{\rm dim}({\rm
Hull}_h(\mathscr{C}^{\bot}_h )),$$ where ${\rm
Hull}_h(\mathscr{C})={\rm
Hull}_h(\mathscr{C}^{\bot_{h}})=\mathscr{C}\cap\mathscr{C}^{\bot_{h}}$.
Obviously, $c={\rm Rank}(HH^{\dag})={\rm Rank}(GG^{\dag})+n-2k$.
\end{theorem}
\par
 For an $[[n,k,d;c]]_{q}$
EAQECC, it is called maximal-entanglement when $c=n-k$.
Refs. \cite{Bowen,Lai1,Lai2,Lu2,Wilde1} have revealed that
maximal-entanglement EAQECCs can both reach the EA-quantum capacity
and EA-hashing bound asymptotically, which can provide higher code
rate and lower SNR (signal to noise ratio). Therefore, it is
worthwhile to exploit how to construct maximal-entanglement EAQECCs
with good performances.
\par
In the following, using the QC code $\mathscr{C}_{q^2}(f,g)$, we
 present construction methods to obtain
maximal-entanglement EAQECCs. By Lemma 1, we know that the Hermitian
dual code of $\mathscr{C}_{q^2}(f,g)$ is generated by the pairs
$(g^{\bot q}(x),0)$ and $(-\bar{{f}^{q}}(x),1)$. Hence, code
$\mathscr{C}_{q^2}(f,g)$ has a parity check matrix as follows
\begin{equation*}
\begin{normalsize}
 H=\left(
\begin{array}{cc}
H_1 &\mathbf{0} \\
 H_{2} &I_{n}  \\
\end{array}
\right),
\end{normalsize}
\end{equation*}
where ${\rm deg}(g(x))\times n$ matrix $H_1$ and $n\times n$ matrix
$H_2$ are respectively circulant matrices determined by ${g}^{\bot
q}(x)$ and $-\bar{f}^q(x)$. $I_{n}$ denotes the $n\times n$ identity
matrix. Let matrix $M$ be the conjugate transpose of the circulant
matrix defined by $f(x)$. One can see that $H_2+M=\mathbf{0}$. In
the rest of the paper, we suppose that ${\rm gcd}(f(x), x^n-1)=1$.
It is easily deduced that matrices $M$ and $H_2$ are invertible.
\begin{theorem}
With the previous notions, suppose that $\mathscr{C}_{q^2}(f,g)$ is
a $[2n,n-{\rm deg}(g(x)),d]_{q^{2}}$ QC code with parity check
matrix
\begin{equation*}
\begin{normalsize}
 H=\left(
\begin{array}{cc}
H_1 &\mathbf{0} \\
 H_{2} &I_{n}  \\
\end{array}
\right),
\end{normalsize}
\end{equation*}
where $H_{1}H^{\dag}_{1}$ be a nonsingular matrix. Define
$P=H_{1}^{\dag}(H_{1}H_{1}^{\dag})^{-1}H_{1}-(H_{2}^{\dag}H_{2})^{-1}$.
If the number 1 isn't in the eigenvalue set of $P$, then there exist
two maximal-entanglement EAQECCs with parameters $[[2n,n-{\rm
deg}(g(x)),d;n+{\rm deg}(g(x))]]_{q}$ and $[[2n,n+{\rm
deg}(g(x)),d^{\bot_{h}};n-{\rm deg}(g(x))]]_{q}$, respectively,
where $d^{\bot_{h}}$ denotes the Hermitian dual distance of
$\mathscr{C}_{q^2}(f,g)$.
\end{theorem}
\begin{proof}
Since $\mathscr{C}_{q^2}(f,g)$ is a $[2n,n-{\rm
deg}(g(x)),d]_{q^{2}}$ linear code, then its Hermitian dual
$\mathscr{C}^{\bot_{h}}_{q^2}(f,g)$ has parameters $[2n,n+{\rm
deg}(g(x)),d^{\bot_{h}}]_{q^{2}}.$ Applying Theorem 5, it provides
two $[[2n,-2{\rm deg}(g(x))+c_{1},d;c_{1}]]_{q}$ and $[[2n,2{\rm
deg}(g(x))+c_{2},d^{\bot_{h}};c_{2}]]_{q}$ EAQECCs, where
$d^{\bot_{h}}$ denotes the Hermitian dual distance of
$\mathscr{C}_{q^2}(f,g)$. Next we compute the number of entangled
ebits $c_{1}$ and $c_{2}$. Note that
\begin{equation*}
\begin{normalsize}
 HH^{\dag}=\left(
\begin{array}{cc}
H_1 &\mathbf{0} \\
 H_{2} &I_{n}  \\
\end{array}
\right)
\end{normalsize}
\begin{normalsize}
\left(
\begin{array}{cc}
H^\dag_1 & H_{2} ^\dag\\
\mathbf{0} &I_{n}  \\
\end{array}
\right)
\end{normalsize}
=
\begin{normalsize}
\left(
\begin{array}{cc}
H_1H^\dag_1 & H_1H_{2} ^\dag \\
H_2H^\dag_1 &H_2H_{2} ^\dag+I_{n}  \\
\end{array}
\right).
\end{normalsize}
\end{equation*}
By the hypothesis, $H_{1}H^{\dag}_{1}$ is a nonsingular matrix, then
we define the following matrices
\begin{equation*}
\begin{normalsize}
A=\left(
\begin{array}{cc}
{(H_{1}H^{\dag}_{1})}^{-1} &\mathbf{0} \\
 0 &I_{n}  \\
\end{array}
\right),
\end{normalsize}
\begin{normalsize}
B=\left(
\begin{array}{cc}
 I_{{\rm deg}(g(x))}&\mathbf{0} \\
-H_{2}H^{\dag}_{1} &I_{n} \\
\end{array}
\right)
\end{normalsize},
\end{equation*}
\begin{equation*}
\begin{normalsize}
 C=\left(
\begin{array}{cc}
 I_{{\rm deg}(g(x))}&\mathbf{0} \\
0 &H_{2}^{-1}  \\
\end{array}
\right)
\end{normalsize},
\begin{normalsize}
D=\left(
\begin{array}{cc}
 I_{{\rm deg}(g(x))}&\mathbf{0} \\
0 &(H_{2} ^\dag)^{-1}\\
\end{array}
\right)
\end{normalsize}.
\end{equation*}
Then
\begin{equation*}
\begin{normalsize}
AHH^{\dag}= \left(
\begin{array}{cc}
I_{{\rm deg}(g(x))} & {(H_{1}H^{\dag}_{1})}^{-1}H_{1}H_{2} ^\dag \\
H_2H^\dag_1 &H_2H_{2} ^\dag+I_{n}  \\
\end{array}
\right),
\end{normalsize}
\end{equation*}
\begin{equation*}
\begin{normalsize}
BAHH^{\dag}= \left(
\begin{array}{cc}
I_{{\rm deg}(g(x))} & {(H_{1}H^{\dag}_{1})}^{-1}H_{1}H_{2} ^\dag \\
\mathbf{0} &-H_{2}H_{1}^{\dag}(H_{1}H_{1}^{\dag})^{-1}H_{1}H_{2}^{\dag}+H_{2}H_{2}^{\dag}+I_{n}  \\
\end{array}
\right),
\end{normalsize}
\end{equation*}
and
\begin{equation*}
\begin{normalsize}
CBAHH^{\dag}D= \left(
\begin{array}{cc}
I_{{\rm deg}(g(x))} & {(H_{1}H^{\dag}_{1})}^{-1}H_{1} \\
\mathbf{0} &-H_{1}^{\dag}(H_{1}H_{1}^{\dag})^{-1}H_{1}+(H_{2}^{\dag}H_{2})^{-1}+I_{n}  \\
\end{array}
\right).
\end{normalsize}
\end{equation*}
If the number 1 isn't an eigenvalue of matrix
$P=H_{1}^{\dag}(H_{1}H_{1}^{\dag})^{-1}H_{1}-(H_{2}^{\dag}H_{2})^{-1}$,
then $-P+I_{n}$ and $CBAHH^{\dag}D$ are both full
 rank matrices.
Note that matrices $A$, $B$, $C$ and $D$ are all invertible, then
$c_{1}= {\rm Rank}(HH^{\dag}) = {\rm Rank}(AHH^{\dag})= {\rm
Rank}(BAHH^{\dag})= {\rm Rank}(CBAHH^{\dag}D)=n+{\rm deg}(g(x)).$
According to Theorem 6, $c_{2}= {\rm Rank}(GG^{\dag})= {\rm
Rank}(HH^{\dag})+2n-2(n+{\rm deg}(g(x)))=n-{\rm deg}(g(x))$. As a
consequence, it provides two maximal-entanglement EAQECCs with
parameters $[[2n,n-{\rm deg}(g(x)),d;n+{\rm deg}(g(x))]]_{q}$ and
$[[2n,n+{\rm deg}(g(x)),d^{\bot_{h}};n-{\rm deg}(g(x))]]_{q}$,
respectively.
\end{proof}
\par
Next we will construct some binary maximal-entanglement EAQECCs with
good parameters according to Theorem 7. Similarly, let $\omega$ be
the primitive element, and elements $0,1,\omega,\omega^2 \in
\mathbb{F}_{4}$ are represented by 0,1,2,3.

\begin{example}
Set $q=2$ and $n=7$. Select the following polynomials in the
quotient ring $\mathbb{F}_{4}[x]/\langle x^{7}-1 \rangle$,
\begin{equation*}
\begin{split}
g(x)=x+1, ~~f(x)=x^5 + 2 x^4 + 3 x^3 + 2 x^2 + 3x.
\end{split}
\end{equation*}
 Then ${{g}^{\bot q}}(x)=x^6 + x^5 + x^4
+ x^3 + x^2 + x + 1$ , $\bar{f}^q(x)=2 x^6 + 3 x^5 + 2 x^4 + 3x^3 +
x^2,$
\begin{equation*}
\begin{small}H_{1}= \left(
\begin{array}{ccccccc}
1111111
\end{array}
\right),~~ H_{2}= \left(
\begin{array}{ccccccc}
  0   0   1 3   2 3   2\\
  2   0   0   1
3 2 3\\
 3  2 0   0   1 3
2\\   2 3 2 0 0   1 3\\
 3
2 3 2 0 0 1\\
   1 3 2 3 2
0 0\\
0 1 3 2 3 2 0
\end{array}
\right)~~
\end{small}
\text{and} ~~
\begin{small} P= \left(
\begin{array}{ccccccc}
 0   0   2 3   2 3   0\\
  0   0   0
2 3 2 3\\
3   0 0   0   2 3 2\\
 2 3   0
0 0 2 3\\
 3 2 3 0 0 0 2\\
2 3   2 3 0 0 0\\
 0 2 3 2
3 0 0\\
\end{array}
\right).
\end{small}
\end{equation*}
 By calculation, $\mathscr{C}_{4}(f,g)$ is a
linear code with parameters $[14,6,7]_{4}$. Since the characteristic
polynomial of matrix $P$ is $x(x^3 + \omega)(x^3 + \omega^2)$, it is
easy to see that the number 1 isn't in its eigenvalue set. Applying
Theorem 7, it can provide a maximal-entanglement EAQECC with
parameters $[[14,6,7;8]]_{2}$. According to Lu's code tables of
maximal-entanglement EAQECCs in \cite{Lu2}, our code is optimal and
better than the best-known $[[14,6,6;8]]_{2}$ EAQECCs. So it breaks
the current records.
\end{example}
\par
\begin{example}
Let $q=2$, $n=11$ and define the following polynomials in the
quotient ring $\mathbb{F}_{4}[x]/\langle x^{11}-1 \rangle$,
\begin{equation*}
\begin{split}
g(x)=x^6 + 3x^5 + 3x^4 + 2x^2 + 2x + 1, ~~f(x)=x^4 + 2x^3 + 3x^2 +
x.
\end{split}
\end{equation*}
 Then ${g}^{\bot q}(x)=x^5 + 3x^4 + x^3 + x^2 + 2x + 1$, $\bar{f}^q(x)=x^{10} + 2x^9 + 3x^8 + x^7.$
We calculate that $\mathscr{C}_{4}(f,g)$ is a $[22,5,13]_{4}$ linear
code
 and its weight enumerator is
\begin{equation*}
\begin{split}
0^1
13^{66}14^{66}15^{198}16^{264}17^{99}18^{132}19^{132}20^{33}21^{33}
\end{split}
\end{equation*}
Using the MacWilliams equation, the weight enumerator of
$\mathscr{C}_{4}^{\bot_{h}}(f,g)$ is given as follows,
\begin{equation*}
\begin{split}
&0^14^{627}5^{6567}6^{52437}7^{364056}8^{2050290}9^{9562740}10^{
37269804}11^{122099016}12^{335494302}\\
&13^{774526170}14^{1493685534}
15^{2389566696}16^{3136710621}17^{3321093204}18^{2767437420}19^{
1748036664} \\
&20^{786523551}21^{224745015}22^{30644469}.
\end{split}
\end{equation*} Hence, $\mathscr{C}_{4}^{\bot_{h}}(f,g)$ is an optimal
linear code with parameters $[22,17,4]_{4}$ and meets requirements
of Theorem 7. Hence, a maximal-entanglement EAQECC with parameters
$[[22,17,4;5]]_{2}$ can be constructed. It has better parameters
than the
 $[[23,17,2;6]]_{2}$ maximal-entanglement EAQECC appeared in
 \cite{Liu}, whose minimum distance does not increase when the entangled
 states are added.
\end{example}
\par
We provide some binary maximal-entanglement EAQECCs in Tables 5 and
6, which are derived from quaternary QC codes $\mathscr{C}_{4}(f,g)$
and $\mathscr{C}^{\bot_{h}}_{4}(f,g)$, respectively. Compared to the
parameters of maximal-entanglement EAQECCs available in \cite{Liu},
 our EAQECCs have better performances.\par

In \cite{Lai1,Lai2}, authors have showed that almost $[[n,k,d;c]]$
EAQECCs were not equivalent to any standard $[[n+c,k,d]]$ QECCs and
had better performances than all $[[n+c,k,d]]$ QECCs. Even if a
maximal-entanglement $[[n,k,d;c]]$ EAQECCs was equivalent to a
$[[n+c,k,d]]$ QECCs, that ebits may be robust against noise when the
ebits were not noiseless. In the last column of Tables 5 and 6, we
list the best known standard QECCs in Grassl's code tables
\cite{Grassl1} with fixed code length $n+c$ and dimension $k$. We
can find that our $[[n,k;c]]$ EAQECCs have grater than or equal to
minimal distances than that of the best-known standard $[[n+c,k,d]]$
QECCs.

\begin{table}
   \caption{EAQECCs from $\mathscr{C}_{4}(f,g)$ and parameters comparison.}
   \begin{center}
   \begin{spacing}{1.6}
   \begin{scriptsize}
   \begin{tabular}{ccccc}
\hline  \hline
   $n$ &$f(x),g(x)$ over $\mathbb{F}_{4}$ &Our EAQECCs& EAQECCs \cite{Liu}&$[[n\!+\!c,k,d]]$\\
   \hline
15&$320213$,\quad$1^30^21^3$ &$[[30,8,15;22]]_{2}$&$[[30,8,7;22]]_{2}$&$[[52,8,10]]_{2}$\\
\hline
17&$1213^201$,\quad$1(10)^2(01)^21$ &$[[34,8,18;26]]_{2}$&$[[34,8,8;26]]_{2}$&$[[60,8,12]]_{2}$\\
\hline
$21$&$1^623201^2$,\quad$1320^4321$ &$[[42,10,17;32]]_{2}$&$[[42,10,8;32]]_{2}$&$[[74,10,14]]_{2}$\\
\hline
$31$&$1^60201$,\quad$10^31^3010^4101^30^31$ & $[[62,10,32;52]]_{2}$& $[[62,10,13;52]]_{2}$&$[[114,10,20]]_{2}$\\
\hline
$35$&$1^523^31$,\quad$12031301203^22^20310212031$ & $[[70,12,37;58]]_{2}$& $[[70,12,14;58]]_{2}$&$[[128,12,22]]_{2}$\\
 \hline
$41$&$1^9232$,\quad$1^30^4(10)^2(01)^20^41^3$ & $[[82,20,33;62]]_{2}$& $[[82,20,13;62]]_{2}$&$[[144,20,23]]_2$\\
\hline \hline
   \end{tabular}
   \end{scriptsize}
   \end{spacing}
\end{center}
\end{table}

\begin{table}
   \caption{EAQECCs from $\mathscr{C}^{\bot_{h}}_{4}(f,g)$ and parameters comparison.}
   \begin{center}
   \begin{spacing}{1.6}
   \begin{scriptsize}
   \begin{tabular}{ccccc}
\hline  \hline
   $n$ &$f(x),g(x)$ over $\mathbb{F}_{4}$ &New EAQECCs& EAQECCs \cite{Liu}&$[[n\!+\!c,k,d]]$\\
   \hline
17&$31^22^2$,\quad$1(10)^2(01)^21$ &$[[34,26,5;8]]_{2}$&$[[34,26,2;8]]_{2}$&$[[42,26,5]]_{2}$\\
\hline
19&$1^52031$,\quad$132^20103^221$ &$[[38,29,5;9]]_{2}$&$[[39,29,2;10]]_{2}$&$[[47,29,5]]_{2}$\\
\hline
$31$&$1^721$,\quad$10^31^3010^4101^30^31$ &$[[62,52,5;10]]_{2}$&$[[62,52,2;10]]_{2}$&$[[72,52,5]]_{2}$\\
\hline \hline
   \end{tabular}
   \end{scriptsize}
   \end{spacing}
\end{center}
\end{table}

\par
 Analogously, we can also present QC extended constructions to
obtain good maximal-entanglement EAQECCs.

\begin{proposition}
Let $q>2$ be a prime power and $\mathscr{C}_{q^2}(f,g)$ be a QC code
with generator matrix $G=(G_1,G_2)$ satisfying the conditions of
Theorem 7. Let $\mathscr{C}_i$ $(i=1,2)$ be a linear code generated
by $G_i$. Choose $x^{(i)}$ to be a codeword in
$\mathscr{C}^{\bot_{h}}_{i}$ and $\alpha_{i}\in
\mathbb{F}_{q^2}^{*}$ such that $\langle x^{(i)},x^{(i)}\rangle _h\neq(p-1)\alpha_{i}^{q+1}$. Then\\
 (i) The code $\mathscr{C}^{'}$ with generator matrix
 \begin{equation*}
 \begin{small}
G^{'}=\left(
 \begin{array}{ccc|ccc|c}
 &     & & &   & & 0\\
 &G_{1}& & &G_{2}&   &  \vdots\\
 &     & & &   & & 0\\
 \hline
 &x^{(1)}&& &0\cdots0&&\alpha_{1}
 \end{array}
\right)
 \end{small}
\end{equation*}
is a  $[2n + 1, n-{\rm deg}(g(x)) + 1]$ linear code with Hermitian
dual distance
$$d^{\bot_{h}}\leq d({\mathscr{C}^{'\bot_{h}}}) \leq d^{\bot_{h}}+1,$$
where $d^{\bot_{h}}$ denotes the Hermitian dual distance of
$\mathscr{C}_{q^2}(f,g)$. Moreover, \\${\rm Rank}(G^{'}{G^{'}}^\dag)
=n-{\rm
deg}(g(x))+1$.\\
 (ii) The code $\mathscr{C}^{''}$ with generator matrix
  \begin{equation*}
   \begin{small}
G^{''}=\left(
 \begin{array}{ccc|ccc|cc}
 &     & & &   & & 0&0\\
 &G_{1}& & &G_{2}&  &  \vdots&\vdots \\
 &     & & &   & & 0&0\\
 \hline
 &x^{(1)}&& &0\cdots0&&\alpha_{1}&0\\
\hline
  &0\cdots0&& &x^{(2)}&&0&\alpha_{2}\\
 \end{array}
\right)
 \end{small}
\end{equation*}
is a  $[2n + 2, n-{\rm deg}(g(x)) +2]$ linear code with Hermitian
dual distance
$$d^{\bot_{h}}\leq d({\mathscr{C}^{''\bot_{h}}}) \leq d^{\bot_{h}}+2,$$
where $d^{\bot_{h}}$ denotes the Hermitian dual distance of
$\mathscr{C}_{q^2}(f,g)$. Moreover, \\${\rm
Rank}(G^{''}{G^{''}}^\dag)=n-{\rm
deg}(g(x))+2$.\\
\end{proposition}
\par
The process of proof is similar to Proposition 1, we omit it here.
From Theorem 5 and Proposition 2, the following result can be
concluded directly.
\begin{theorem}
Let $q>2$ be a prime power and $\mathscr{C}_{q^2}(f,g)$ a QC code
with generator matrix $G=(G_1,G_2)$ satisfying the conditions of
Theorem 7. Then, there exist two maximal-entanglement EAQECCs with
parameters $[[2n+1,n+{\rm deg}(g(x)),d^{'};n-{\rm deg}(g(x))+1]]$
and $[[2n+2,n+{\rm deg}(g(x)),d^{''};n-{\rm deg}(g(x))+2]]$,
respectively. Moreover, $d^{\bot_{h}}\leq
d({\mathscr{C}^{'\bot_{h}}}) \leq d^{\bot_{h}}+1$ and
$d^{\bot_{h}}\leq d({\mathscr{C}^{''\bot_{h}}}) \leq
d^{\bot_{h}}+2$, where $d^{\bot_{h}}$ denotes the Hermitian dual
distance of $\mathscr{C}_{q^2}(f,g)$.
\end{theorem}
\begin{example}
Assume that $q=9$ and $n=10$. Let $\zeta$ be a primitive element of
$\mathbb{F}_{81}$. Consider the following polynomials in
$\mathbb{F}_{81}[x]/\langle x^{10}-1\rangle$,
\begin{equation*}
\begin{split}
g(x)=x^7 + \zeta^{44}x^6 + \zeta^{58}x^5 + \zeta^{52}x^4 +
\zeta^{36}x^3 + \zeta^{10}x^2 + \zeta^{44}x + \zeta^{48},
\end{split}
\end{equation*}
$$f(x)=\zeta^{14}x^2+\zeta^2x+1.$$
 Then ${g}^{\bot q}(x)=\zeta^8x^3 + \zeta^{12}x^2 + \zeta^{36}x + 1$, $\bar{f}^q(x)=\zeta^{18}x^9 + \zeta^{46}x^8 + 1,$
\begin{equation*}
\begin{tiny}
H_{1}= \left(
\begin{array}{cccccccccc}
 1 \zeta^{36} \zeta^{12}  \zeta^8    0    0    0    0    0    0\\
0 1 \zeta^{36} \zeta^{12}  \zeta^8    0    0    0    0    0\\
0 0 1 \zeta^{36} \zeta^{12} \zeta^8 0    0    0    0\\
0    0 0 1 \zeta^{36} \zeta^{12} \zeta^8 0 0 0\\
0    0    0 0 1 \zeta^{36} \zeta^{12} \zeta^8 0 0\\
 0 0 0 0 0  1 \zeta^{36} \zeta^{12} \zeta^8 0\\
 0 0 0 0 0 0 1 \zeta^{36} \zeta^{12}  \zeta^8 \\
\end{array}
\right),~~ H_{2}= \left(
\begin{array}{ccccccc}
2    0    0    0    0    0    0    0  \zeta^6 \zeta^{58}\\
\zeta^{58}    2    0    0    0    0    0    0    0  \zeta^6\\
\zeta^6 \zeta^{58} 2 0 0    0    0    0    0    0\\
0  \zeta^6 \zeta^{58} 2 0 0 0 0 0    0 \\
0    0  \zeta^6 \zeta^{58}    2 0 0 0 0 0\\
 0 0 0 \zeta^6 \zeta^58    2 00 0 0\\
 0  0 0 0 \zeta^6 \zeta^{58} 2 0    0    0\\
 0    0    0    0 0 \zeta^6 \zeta^{58} 2 0 0\\
0 0 0 0 0    0  \zeta^6 \zeta^{58} 2   0\\
 0    0 0 0 0 0 0 \zeta^6 \zeta^{58} 2\\
\end{array}
\right),
\end{tiny}
\begin{tiny} P= \left(
\begin{array}{cccccccccc}
 \zeta^{60} \zeta^{19}  1 \zeta^{28} \zeta^{41}  0 \zeta^{49} \zeta^{12} 1
 \zeta^{11}\\
\zeta^{11} \zeta^{60} \zeta^{19}    1 \zeta^{28} \zeta^{41}    0
\zeta^{49} \zeta^{12} 1\\ 1 \zeta^{11} \zeta^{60} \zeta^{19}    1
\zeta^{28} \zeta^{41} 0 \zeta^{49} \zeta^{12}\\ \zeta^{12} 1
\zeta^{11} \zeta^{60} \zeta^{19} 1 \zeta^{28} \zeta^{41} 0 \zeta^{49}\\
\zeta^{49} \zeta^{12} 1 \zeta^{11} \zeta^{60} \zeta^{19} 1
\zeta^{28} \zeta^{41} 0\\  0
\zeta^{49} \zeta^{12} 1 \zeta^{11} \zeta^{60} \zeta^{19} 1 \zeta^{28} \zeta^{41}\\
\zeta^{41}
0 \zeta^{49} \zeta^{12} 1 \zeta^{11} \zeta^{60} \zeta^{19} 1 \zeta^{28}\\
\zeta^{28} \zeta^{41} 0 \zeta^{49} \zeta^{12} 1 \zeta^{11}
\zeta^{60} \zeta^{19} 1\\  1 \zeta^{28}
\zeta^{41}    0 \zeta^{49} \zeta^{12} 1 \zeta^{11} \zeta^{60} \zeta^{19}\\
\zeta^{19} 1 \zeta^{28}
\zeta^{41} 0 \zeta^{49} \zeta^{12} 1 \zeta^{11} \zeta^{60}\\
\end{array}
\right).
\end{tiny}
\end{equation*}
 Since the characteristic polynomial of matrix $P$ is
$x(x+\zeta^{10})(x+\zeta^{30})(x+\zeta^{50})(x+
\zeta^{70})(x+\zeta^{60})^2(x+\zeta^{20})^3$, then the number 1
isn't an eigenvalue of the matrix $P$. Moreover,
$H_{1}H^{\dag}_{1}$ is nonsingular. Therefore, the QC code
$\mathscr{C}_{q^2}(f,g)$ with a generator matrix
\begin{equation*}
 \begin{small}
 G= \left(
\begin{array}{cccccccccc cccccccccc}
\zeta^{48} \zeta^{44} \zeta^{10} \zeta^{36} \zeta^{52} \zeta^{58}
\zeta^{44} 1 0 0 \zeta^{48} \zeta^{58} \zeta^{18} \zeta^{15}
\zeta^{36} \zeta^{56} \zeta^{24}
    \zeta^{23} \zeta^{65} \zeta^{14}\\
0 \zeta^{48} \zeta^{44} \zeta^{10} \zeta^{36} \zeta^{52} \zeta^{58}
\zeta^{44} 1 0 \zeta^{14} \zeta^{48} \zeta^{58} \zeta^{18}
\zeta^{15} \zeta^{36} \zeta^{56}
    \zeta^{24} \zeta^{23} \zeta^{65}\\
0 0 \zeta^{48} \zeta^{44} \zeta^{10} \zeta^{36} \zeta^{52}
\zeta^{58} \zeta^{44} 1 \zeta^{65} \zeta^{14} \zeta^{48} \zeta^{58}
\zeta^{18} \zeta^{15} \zeta^{36}
    \zeta^{56} \zeta^{24} \zeta^{23}\\
\end{array}
\right)
 \end{small}
\end{equation*}
satisfies the conditions of Theorem 7.
\par
We choose $x^{(1)}=(\zeta^{44}, \zeta^{71}, \zeta^{56}, \zeta^{22},
\zeta^{52}, \zeta^{73}, \zeta^{33}, \zeta^{58}, \zeta^{58},
\zeta^{33})$ and $x^{(2)}= (\zeta^{18}, \zeta^{41}, 2, \zeta^{10},
\zeta^{17}, \zeta^{31}, \zeta^{71}, \zeta^{61}, \zeta^{66},
\zeta^{75})$. From $x^{(1)}{x^{(1)}}^{\dag}=\zeta^{60}$ and
$x^{(2)}{x^{(2)}}^{\dag}=\zeta^{50}$, we may take $\alpha_{1}$ and
$\alpha_{2}$ both equal to $1$. Hence, the extended QC code
$\mathscr{C}^{''}$ has a generator matrix as follows
\begin{equation*}
 \begin{small}
 G^{''}= \left(
\begin{array}{cccccccccc cccccccccc cc}
\zeta^{48} \zeta^{44} \zeta^{10} \zeta^{36} \zeta^{52} \zeta^{58}
\zeta^{44} ~1 ~0~ 0~ \zeta^{48} \zeta^{58} \zeta^{18} \zeta^{15}
\zeta^{36} \zeta^{56} \zeta^{24}
    \zeta^{23} \zeta^{65} \zeta^{14}~0~0\\
0~ \zeta^{48} \zeta^{44} \zeta^{10} \zeta^{36} \zeta^{52} \zeta^{58}
\zeta^{44} ~1~ 0~ \zeta^{14} \zeta^{48} \zeta^{58} \zeta^{18}
\zeta^{15} \zeta^{36} \zeta^{56}
    \zeta^{24} \zeta^{23} \zeta^{65}~0~0\\
0 ~0 \zeta^{48} \zeta^{44} \zeta^{10} \zeta^{36} \zeta^{52}
\zeta^{58} \zeta^{44} ~1 ~\zeta^{65} \zeta^{14} \zeta^{48}
\zeta^{58} \zeta^{18} \zeta^{15} \zeta^{36}
    \zeta^{56} \zeta^{24} \zeta^{23}~0~0\\
\zeta^{44} \zeta^{71} \zeta^{56} \zeta^{22} \zeta^{52} \zeta^{73}
\zeta^{33} \zeta^{58} \zeta^{58} \zeta^{33} ~~0~~ 0~~ 0~~ 0~~ 0~~
0~~ 0~~ 0~~ 0~~ 0~~ 1~~ 0\\
0 ~~0 ~~0~~ 0~~ 0~~ 0~~ 0~~ 0~~ 0~~ 0~~ \zeta^{18} \zeta^{41} ~2~
\zeta^{10} \zeta^{17} \zeta^{31} \zeta^{71} \zeta^{61}
\zeta^{66} \zeta^{75} ~~0~~ 1\\
\end{array}
\right).
 \end{small}
\end{equation*}
Applying the MacWilliams equation, we calculate that
${\mathscr{C}^{''\bot_{h}}}$ is a $[22,17,5]_{81}$ linear code.
According to Theorem 8, a new maximal-entanglement EAQECC with
parameters $[[22,17,5;5]]_{9}$ can be derived, which is superior to
the codes with parameters $[[23,17,3;6]]_{9}$ appeared in
\cite{Liu}. Note that a standard pure $[[22,10,5]]_{9}$ QECC is the
best code meeting the quantum GV bounds with code length $n=22$ and
minimum distance $d=5$. Compared with this standard QECC, our
constructed EAQECC has $9^7=4782969$ more codewords for the same
code length and minimum distance although we add $5$ entanglement
ebits indeed.

\end{example}

\section {Conclusions}
In this paper, by a class of one-generator QC codes, we presented QC
extended constructions that preserved the self-orthogonality. As an
application, some good stabilizer QECCs over small finite fields
$\mathbb{F}_{2}$ and $\mathbb{F}_{3}$ were obtained. In the binary
case, some of our quantum codes broken or matched the current
records. In the ternary case, our codes filled some gaps or had
better performances than the current results.\par

It is well-known that the most common way of constructing QECCs now
is from cyclic and constacyclic codes \cite{Aly,Kai,Ma,Zhu}. But in
most cases, in order to gain good QECCs, we need the code length $n$
to divide ${q^s-1}$ for some positive integer $s$. From our extended
QC constructions, one can see that our method can breakthrough the
restriction partly, which produces QECCs with more flexible code
lengths. Further, we have constructed maximal-entanglement EAQECCs
from
 QC codes and their extended codes as well. Some good maximal-entanglement EAQECCs were derived and their
parameters were compared. To the best of our knowledge,
 this is the first attempt to construct maximal-entanglement EAQECCs from QC codes and their extended codes.
 \par
However, one can find that our construction only can provide QECCs
and EAQECCs with a relatively small distance. As the dimension
increase, calculating the exact Hermitian dual distance will be
computationally intractable (NP-hard) even if we used the
MacWilliams equation. So in future study, a lower bound for our QC
extended construction is extremely valuable.

\par

\section*{Acknowledgments}
This work is supported by National Natural Science Foundation of
China (Nos.11471011, 11801564, 11901579).



\end{document}